\documentclass{acm_proc_article-sp}

\usepackage{algorithmic}
\usepackage{algorithm}
\usepackage{graphicx}
\usepackage{float}
\usepackage{subfloat}
\usepackage{subfig}
\usepackage{url}
\usepackage{amssymb}
\usepackage{yfonts}
\usepackage{eufrak}

\usepackage[table]{xcolor}
\xdefinecolor{gray95}{gray}{0.65}
\xdefinecolor{gray25}{gray}{0.8}

\usepackage{amsfonts}
\usepackage{mathtools}
\usepackage{cite}
\usepackage{multirow}
\usepackage{array}
\usepackage{multicol}

\usepackage{amsthm}
\usepackage{xspace}

\usepackage{amsmath}

\usepackage[table]{xcolor}
\xdefinecolor{gray95}{gray}{0.65}
\xdefinecolor{gray25}{gray}{0.8}

\DeclareCaptionType{copyrightbox}
\newtheoremstyle{exampstyle}
{1pt} 
{1pt} 
{} 
{} 
{\bfseries} 
{.} 
{.5em} 
{} 

\theoremstyle{exampstyle}

\newtheorem{metric}{Metric}

\newcommand{\nfp}{19\xspace}
\newcommand{\gplts}{418\xspace}

\newtheorem{theorem}{Theorem}

\newtheorem{definition}{Definition}

\begin{document}

%

\title{Towards a Benchmark and a Comparison Framework
\\ for Combinatorial Interaction Testing of \\
Software Product Lines}

\numberofauthors{6} 
%
\author{
%
%
\alignauthor
Roberto E. Lopez-Herrejon\\
       \affaddr{Johannes Kepler University Linz, Austria}\\
       \email{rlopez@jku.at}
\alignauthor
Javier Ferrer\\
       \affaddr{University of M\'alaga, Spain}\\
       \email{ferrer@lcc.uma.es}
\alignauthor
Francisco Chicano\\
       \affaddr{University of M\'alaga, Spain}\\
       \email{chicano@lcc.uma.es}
\and  
\alignauthor Evelyn Nicole Haslinger\\
       \affaddr{Johannes Kepler University Linz, Austria}\\
       \email{evelyn.haslinger@jku.at}
\alignauthor Alexander Egyed\\
       \affaddr{Johannes Kepler University Linz, Austria}\\
       \email{alexander.egyed@jku.at}
\alignauthor Enrique Alba\\
       \affaddr{University of M\'alaga, Spain}\\
       \email{eat@lcc.uma.es}
}

\maketitle

\begin{abstract}
As Software Product Lines (SPLs) are becoming a more pervasive development practice, their effective testing is becoming a more important concern.
In the past few years many SPL testing approaches have been proposed, among them, are those that support Combinatorial Interaction Testing (CIT) whose premise is to select a group of products where faults, due to feature interactions, are more likely to occur. Many CIT techniques for SPL testing have been put forward; however, no systematic and comprehensive comparison among them has been performed. To achieve such goal two items are important: a common benchmark of feature models, and an adequate comparison framework. In this research-in-progress paper, we propose 19 feature models as the base of a benchmark, which we apply to three different techniques in order to analyze the comparison framework proposed by Perrouin et al. We identify the shortcomings of this framework and elaborate alternatives for further study.

\end{abstract}





\section{Introduction}
\label{sec:introduction}

A \textit{Software Product Line (SPL)} is a family of related software systems, which provide different feature combinations \cite{SPLE}. The effective management and realization of \textit{variability}  \--- the capacity of software artifacts to vary \cite{DBLP:journals/spe/SvahnbergGB05} \--- can lead to substantial benefits such as increased software reuse, faster product customization, and reduced time to market.

~\linebreak
Systems are being built, more and more frequently, as SPLs rather than individual products because of several technological and marketing trends.
This fact has created an increasing need for testing approaches that are capable of coping with large numbers of feature combinations that characterize SPLs. Many testing alternatives have been put forward~\cite{DBLP:journals/infsof/EngstromR11, DBLP:journals/infsof/NetoMMAM11,  DBLP:conf/splc/LeeKL12,DBLP:journals/sigsoft/MachadoMA12}. 
Salient among them are those that support \textit{Combinatorial Interaction Testing (CIT)} whose premise is to select a group of products where faults, due to feature interactions, are more likely to occur. In particular, most of the work has focused on pairwise testing whereby the interactions of two features are considered~\cite{DBLP:conf/icst/PerrouinSKBT10, DBLP:conf/splc/OsterMR10, GCD11, DBLP:conf/issre/HervieuBG11, DBLP:journals/sqj/LochauOGS12, DBLP:conf/models/CichosOLS11, DBLP:conf/caise/EnsanBG12}.
With all these pairwise testing approaches available the question now is: how do they compare?
To answer this question, two items are necessary: a common benchmark of feature models, and an adequate comparison framework.
In this research-in-progress paper, we propose a set of 19 feature models as a base of the comparison benchmark. In addition, we use these feature models to illustrate Perrouin et al.'s comparison framework~\cite{DBLP:journals/sqj/PerrouinOSKBT12}. We identify some shortcomings, and elaborate alternatives for further study.

The organization of the paper is as follows. In Section~\ref{sec:feature-models} we present the basic background on feature models. Section~\ref{sec:combinatorialtesting} describes the basic terminology of CIT and how it is applied to SPLs. Section~\ref{sec:benchmark} presents the list of feature models that we proposed as basic benchmark.
Section~\ref{sec:framework} summarizes and illustrates Perrouin et al.'s comparison framework. Section~\ref{sec:algorithms} sketches the three CIT algorithms used to illustrate both the benchmark and the comparison framework.
Section~\ref{sec:Evaluation} presents the results of our evaluation and its analysis. Section~\ref{sec:relatedwork} briefly summarizes the related work and Section~\ref{sec:conclusions} outlines the conclusions and future work.


\section{Feature Models and Running Example}
\label{sec:feature-models}

Feature models have become a \textit{de facto} standard for modelling the common and variable features of an SPL and their relationships collectively forming a tree-like structure. The nodes of the tree are the features, which are depicted as labelled boxes, and the edges represent the relationships among them.
Thus, a feature model denotes the set of feature combinations that the products of an SPL can have~\cite{KCH+90}.

Figure~\ref{fig:gpl-fm} shows the feature model of our running example, the \textit{Graph Product Line (GPL)}, a standard SPL of basic graph algorithms that has been extensively used as a case study in the product line community \cite{DBLP:conf/gcse/Lopez-HerrejonB01}.
A product has feature \texttt{GPL} (the root of the feature model) which contains its core functionality, and a driver program (\texttt{Driver}) that sets up the graph examples (\texttt{Benchmark}) to which a combination of graph algorithms (\texttt{Algorithms}) are applied. The types of graphs (\texttt{GraphType}) can be either directed (\texttt{Directed}) or undirected (\texttt{Undirected}), and can optionally have weights (\texttt{Weight}). Two graph traversal algorithms (\texttt{Search}) are available: either Depth First Search (\texttt{DFS}) or Breadth First Search (\texttt{BFS}). A product must provide at least one of the following algorithms: numbering of nodes in the traversal order (\texttt{Num}), connected components (\texttt{CC}), strongly connected components (\texttt{SCC}), cycle checking (\texttt{Cycle}), shortest path (\texttt{Shortest}), minimum spanning trees with Prim's algorithm (\texttt{Prim}) or Kruskal's algorithm (\texttt{Kruskal}).

\begin{figure}
\includegraphics[scale=0.3]{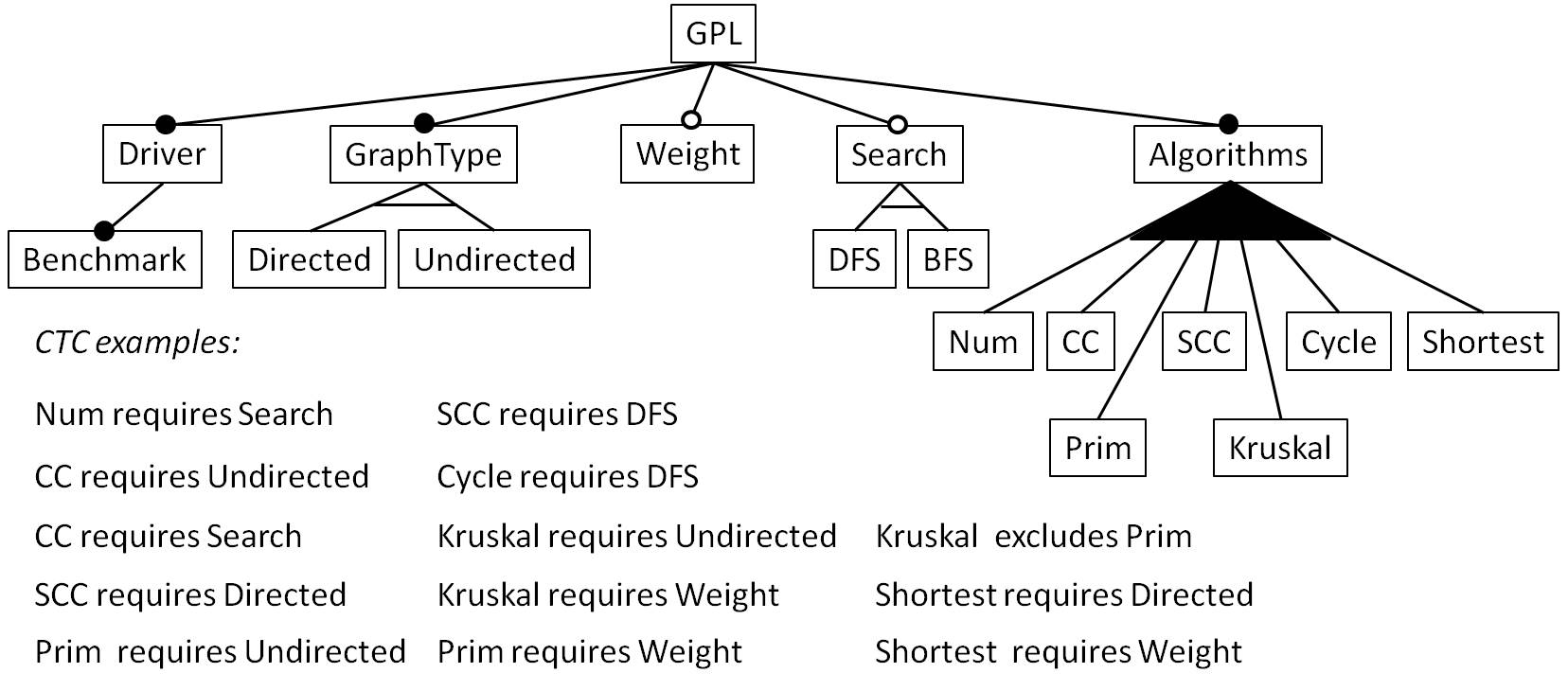}
\caption{Graph Product Line Feature Model}%
\label{fig:gpl-fm}
\end{figure}
In a feature model, each feature (except the root) has one parent feature and can have a set of child features. Notice here that a child feature can only be included in a feature combination of a valid product if its parent is included as well. The root feature is always included. There are four kinds of feature relationships:
\begin{itemize}
	\item \textit{Mandatory features} are depicted with a filled circle. A mandatory feature is selected whenever its respective parent feature is selected.  For example, features \texttt{Driver} and \texttt{GraphType}.
	
	\item \textit{Optional features} are depicted with an empty circle. An optional feature may or may not be selected if its respective parent feature is selected. An example is feature \texttt{Weight}.
	
	\item \textit{Exclusive-or relations} are depicted as empty arcs crossing over a set of lines connecting a parent feature with its child features. They indicate that exactly one of the features in the exclusive-or group must be selected whenever the parent feature is selected. For example, if feature \texttt{Search} is selected, then either feature \texttt{DFS} or feature \texttt{BFS} must be selected.
	
	\item \textit{Inclusive-or relations} are depicted as filled arcs crossing over a set of lines  connecting a parent feature with its child features. They indicate that at least one of the features in the inclusive-or group must be selected if the parent is selected. If for instance, feature \texttt{Algorithms} is selected then at least one of the features \texttt{Num}, \texttt{CC}, \texttt{SCC}, \texttt{Cycle}, \texttt{Shortest}, \texttt{Prim}, or \texttt{Kruskal} must be selected.
	
\end{itemize}
Besides the parent-child relations, features can also relate across different branches of the feature model with \textit{Cross-Tree Constraints (CTCs)}. Figure \ref{fig:gpl-fm} shows the CTCs of GPL.
For instance, \texttt{Cycle requires DFS} means that whenever feature \texttt{Cycle} is selected, feature \texttt{DFS} must also be selected. As another example, \texttt{Prim excludes Kruskal}  means that both features cannot be selected at the same time in any product.
These constraints as well as those implied by the hierarchical relations between features are usually expressed and checked using propositional logic, for further details refer to \cite{DBLP:journals/is/BenavidesSC10}. Now we present the basic definitions on which SPL testing terminology is defined in the next section.

\begin{table*}
\caption{Sample Feature Sets of GPL}
\center
\small
\begin{tabular}{|c|l|l|l|l|l|l|l|l|l|l|l|l|l|l|l|l|l|l|l|}
\hline
FS & GPL  & Dri & Gtp & W & Se & Alg & B & D & U & DFS & BFS & N  & CC & SCC & Cyc & Sh & Prim & Kru \tabularnewline
\hline
fs0 & \checkmark & \checkmark & \checkmark & \checkmark &  & \checkmark & \checkmark &  & \checkmark &  &  &  &  &  &  &  & \checkmark & \tabularnewline
\hline
fs1 & \checkmark & \checkmark & \checkmark & \checkmark & \checkmark & \checkmark & \checkmark &  & \checkmark & \checkmark &  &  & \checkmark &  &  &  &  & \checkmark \tabularnewline
\hline
fs2 & \checkmark & \checkmark & \checkmark &  & \checkmark & \checkmark & \checkmark & \checkmark &  & \checkmark &  & \checkmark &  &  & \checkmark &  &  & \tabularnewline
\hline
fs3 & \checkmark & \checkmark & \checkmark & \checkmark & \checkmark & \checkmark & \checkmark & \checkmark &  &  & \checkmark & \checkmark &  &  &  & \checkmark & & \tabularnewline
\hline
fs4 & \checkmark & \checkmark & \checkmark & \checkmark & \checkmark & \checkmark & \checkmark & \checkmark &  & \checkmark &  & \checkmark &  & \checkmark & \checkmark & \checkmark &  & \tabularnewline
\hline
fs5 & \checkmark & \checkmark & \checkmark & \checkmark & \checkmark & \checkmark & \checkmark &  & \checkmark & \checkmark &  & \checkmark & \checkmark &  & \checkmark &  & \checkmark & \tabularnewline
\hline
fs6 & \checkmark & \checkmark & \checkmark & \checkmark & \checkmark & \checkmark & \checkmark &  & \checkmark &  & \checkmark & \checkmark & \checkmark &  &  &  &  & \checkmark  \tabularnewline
\hline
fs7 & \checkmark & \checkmark & \checkmark & \checkmark & \checkmark & \checkmark & \checkmark &  & \checkmark & \checkmark &  & \checkmark & \checkmark &  & \checkmark &  &  & \tabularnewline

\hline
\end{tabular}
\\ Driver (Dri), GraphType (Gtp), Weight (W), Search (Se), Algorithms (Alg), Benchmark (B), Directed (D), Undirected (U), Num (N), Cycle (Cyc), Shortest (Sh), Kruskal (Kr).
\label{tab:gpl-featuresets}
\end{table*}

\begin{definition} \textit{Feature List (FL) is the list of features in a feature model.}
\end{definition}

\begin{sloppypar}
The \texttt{FL} for the GPL feature model is \texttt{[GPL, Driver, Benchmark, GraphType, Directed, Undirected, Weight, Search, DFS, BFS, Algorithms, Num, CC, SCC, Cycle, Shortest, Prim, Kruskal]}.
\end{sloppypar}

\begin{definition} \textit{A feature set, also called product in an SPL, is a 2-tuple {[}sel,$\overline{sel}${]}
where sel and $\overline{sel}$ are respectively the set of selected
and not-selected features of a member product\footnote{Definition based on \cite{DBLP:journals/is/BenavidesSC10}.}. Let FL be a feature
list, thus sel, $\overline{sel}$ $\subseteq$ FL, sel $\cap$ $\overline{sel}$
= $\emptyset$, and sel $\cup$ $\overline{sel}$ = FL. The terms
p.sel and p.$\overline{sel}$ respectively refer to the set of selected
and not-selected features of product p.}
\end{definition}

\begin{definition} \textit{A feature set \texttt{fs} is valid in feature model \texttt{fm}, iff \texttt{fs} does not contradict any of the constraints introduced by \texttt{fm}. We will denote with $FS$ the set of valid feature sets for a feature model (we omit the feature model in the notation).}
\end{definition}

\begin{sloppypar}
For example, the feature set \texttt{fs0={[}\{{GPL, Driver, GraphType, Weight, Algorithms, Benchmark,  Undirected, Prim}\}, \{Search, Directed, DFS, BFS, Num, CC, SCC, Cycle, Shortest, Kruskal\}{]}} is valid.
As another example, a feature set with features \texttt{DFS} and \texttt{BFS}
would not be valid because it violates the constraint of the exclusive-or relation which establishes that these two features cannot appear selected together in the same feature set. The GPL feature model denotes 73 valid feature sets, some of them depicted in Table~\ref{tab:gpl-featuresets}, where selected features are ticked (\checkmark) and unselected features are empty. 
\end{sloppypar}


\begin{definition} \textit{A feature \texttt{f} is a core feature if it is selected in all the valid feature sets of a feature model \texttt{fm}, and is a variant feature if it is selected in some of the feature sets.}
\end{definition}

\begin{sloppypar}
For example \texttt{GPL}, \texttt{Driver}, \texttt{Benchmark}, \texttt{GraphType} and \texttt{Algorithms} are core features and the remaining ones are variant features.
\end{sloppypar}


\section{Combinatorial Interaction Testing for Software Product Lines}
\label{sec:combinatorialtesting}
Combinatorial Interaction Testing (CIT) is a testing approach that constructs samples to drive the systematic testing of software system configurations \cite{DBLP:journals/tse/CohenDS08}. When applied to SPL testing, the idea is to select 
a representative subset of products where interaction errors are more likely to occur rather than testing the complete product family \cite{DBLP:journals/tse/CohenDS08}. 
In this section we provide the basic terminology of CIT within the context of SPLs.


\begin{definition} \textit{A t-set ts is a 2-tuple {[}sel,$\overline{sel}${]} representing a partially configured product, defining the selection of t features of the feature list FL, i.e. $ts.sel \cup ts.\overline{sel} \subseteq FL $ $\land$ $ts.sel \cap ts.\overline{sel}=\emptyset$ $\land$ $|ts.sel \cup ts.\overline{sel}| = t $. We say t-set ts is covered by feature set fs iff $ts.sel \subseteq fs.sel $ $\land$ $ ts.\overline{sel} \subseteq fs.\overline{sel}$.}
\end{definition}  

\begin{definition}  \textit{A $t$-set \texttt{ts} is valid in a feature model \texttt{fm} if there exists a valid feature set \texttt{fs} that covers \texttt{ts}.  The set of all valid $t$-sets for a feature model is denoted with $TS$\footnote{We also omit here the feature model in the notation for the sake of clarity.}.}
\end{definition}

\begin{definition} \textit{A t-wise covering array tCA for a feature model \texttt{fm} is a set of valid feature sets that covers all valid t-sets denoted by \texttt{fm}\footnote{Definition based on  \cite{JHF12}.}. We also use the term test suite to refer to a covering array.} 
\end{definition}

Let us illustrate these concepts for pairwise testing, meaning \texttt{t=2}. From the feature model in Figure \ref{fig:gpl-fm}, a valid 2-set is \texttt{[\{Driver\},\{Prim\}]}. It is valid because the selection of feature \texttt{Driver} and the non-selection of feature \texttt{Prim} do not violate any constraints. As another example, the 2-set \texttt{[\{Kruskal,DFS\}, $\emptyset$]} is valid because there is at least one feature set, for instance \texttt{fs1} in Table \ref{tab:gpl-featuresets}, where both features are selected. The 2-set \texttt{[$\emptyset$, \{SCC,CC\}]} is also valid because there are valid feature sets that do not have any of these features selected, for instance feature sets \texttt{fs0}, \texttt{fs1}, and \texttt{fs3}. Notice however that the 2-set \texttt{[$\emptyset$,\{Directed, Undirected\}]} is not valid. This is because feature \texttt{GraphType} is present in all the feature sets (mandatory child of the root) so either \texttt{Directed} or \texttt{Undirected} must be selected. In total, our running example has  $\gplts$ valid 2-sets.

Based on Table \ref{tab:gpl-featuresets}, the three valid 2-sets just mentioned above are covered as follows. The 2-set \texttt{[\{Driver\},\{Prim\}]} is covered by feature sets \texttt{fs1}, \texttt{fs2}, \texttt{fs3}, \texttt{fs4}, \texttt{fs6}, and \texttt{fs7}. Similarly, the 2-set \texttt{[\{Kruskal,DFS\}, $\emptyset$]} is covered by feature set \texttt{fs1}, and  \texttt{[$\emptyset$, \{SCC,CC\}]} is covered by feature sets \texttt{fs0}, \texttt{fs2}, and \texttt{fs3}.


\section{Basic Benchmark}
\label{sec:benchmark}


We propose the use of 19 \emph{realistic} feature models as a basis for a comparison benchmark. By realistic we mean that these models meet three basic requirements:
\begin{enumerate}

\item \textit{Available Source Code}. Because the ultimate goal of this line of research is to evaluate the effectiveness of the testing approaches, it is thus of the utmost importance that the source code associated to the proposed feature models be available in a complete form, although perhaps not be thoroughly documented.
%

\item \textit{Explicit Feature Model}. We consider feature models that are explicitly provided by the SPL authors. This requirement is to prevent any misunderstandings or omissions that any techniques to reverse engineering feature models from other artifacts can potentially have.

\item  \textit{Plausible number of products}. It does not take many features to create feature models with a huge number of potential products. We arbitrarily chose two million as the maximum number of products denoted by the feature models in the benchmark. We would argue this is a reasonable number of products that a large company or open source community could potentially maintain and most importantly thoroughly test.

\end{enumerate}

In order to find the feature models that meet these criteria we searched  proceedings from SPL-related venues such as SPLC, Vamos, ICSE, ASE, and FSE published over the last five years. In addition, we consulted the following websites and repositories: SPL Conqueror~\cite{DBLP:journals/infsof/SiegmundRKGAK13}, FeatureHouse~\cite{FeatureHouse}, SPL2go~\cite{SPL2GO}, and SPLOT~\cite{SPLOT}\footnote{Search performed during August-September 2013.}. Table~\ref{tab:feature-models} summarizes the feature models used in our evaluation. It shows the number of features, number of products, and their application domain with the reference from where they were obtained.

We should point out that some of the pairwise testing approaches identified and mentioned in Section~\ref{sec:relatedwork} already use some examples from the SPLOT website; however, to the best of our knowledge, the criteria for their   feature models selection is not specified precisely. In our experience with this repository, based on the information provided by the model authors on the SPLOT website itself, we either were not able to trace the code sources of the feature models or we found semantic mistakes in them.

We should stress that this list of feature models is by no means complete. Our expectation, as a result of this paper, is that the SPL community proposes new feature models to add or remove to this benchmark, perhaps filling in details that were not found by our search, and adding or refining our selection criteria.

\begin{table}  
\centering
\scriptsize
\begin{tabular}{|l|r|r|l|}
\hline
Feature Model & NF  & NP &  Domain\\ \hline
Apache  & 10  & 256 & web server~\cite{DBLP:journals/infsof/SiegmundRKGAK13}\\ \hline
argo-uml-spl  & 11  & 192 &  UML tool~\cite{ArgoUML}\\ \hline
BDB*  & 117  & 32 &  database~\cite{FeatureHouse}\\ \hline  
BDBFootprint  & 9  & 256  &  database~\cite{DBLP:journals/infsof/SiegmundRKGAK13}\\ \hline
BDBMemory  & 19  & 3,840 &  database~\cite{DBLP:journals/infsof/SiegmundRKGAK13}\\ \hline
BDBPerformance  & 27  & 1,440 &  database~\cite{DBLP:journals/infsof/SiegmundRKGAK13} \\ \hline
Curl  & 14  & 1024 & data trasfer~\cite{DBLP:journals/infsof/SiegmundRKGAK13}\\ \hline
DesktopSearcher  & 22  & 462 & file search~\cite{SPL2GO}\\ \hline
fame$\_$dbms$\_$fm  & 20  & 320 & database~\cite{SPL2GO} \\ \hline
gpl  & 18  & 73 & graph algorithms~\cite{DBLP:conf/gcse/Lopez-HerrejonB01}\\ \hline
LinkedList  & 27  & 1,344 & data structures~\cite{DBLP:journals/infsof/SiegmundRKGAK13}\\ \hline
LLVM  & 12  & 1,024 & compiler library~\cite{DBLP:journals/infsof/SiegmundRKGAK13}\\ \hline
PKJab  & 12  & 72 & messenger~\cite{DBLP:journals/infsof/SiegmundRKGAK13} \\ \hline
Prevayler  & 6  & 32 & object persistence~\cite{DBLP:journals/infsof/SiegmundRKGAK13} \\ \hline
SensorNetwork  & 27  & 16,704 &  networking~\cite{DBLP:journals/infsof/SiegmundRKGAK13} \\ \hline
TankWar  & 37  & 1,741,824 & game~\cite{FeatureHouse}\\ \hline
Wget  & 17  & 8,192 & file retrieval~\cite{DBLP:journals/infsof/SiegmundRKGAK13}\\ \hline
x264  & 17  & 2,048 & video encoding~\cite{DBLP:journals/infsof/SiegmundRKGAK13}\\ \hline
ZipMe  & 8  & 64 & data compression~\cite{DBLP:journals/infsof/SiegmundRKGAK13}\\ \hline
\end{tabular}
\\
\textbf{NF}: Number of Features, \textbf{NP}:Number of Products, 
\\ *\textbf{BDB} prefix standards for Berkeley database. 		
\caption{Feature Models Summary}
\label{tab:feature-models}
\end{table}


\section{Comparison Framework}
\label{sec:framework}

In this section we present the four metrics that constitute the framework proposed by Perrouin \emph{et al.} for the comparison of pairwise testing approaches for  SPLs \cite{DBLP:journals/sqj/PerrouinOSKBT12}. We define them based on the terminology presented in Sections~\ref{sec:feature-models} and~\ref{sec:combinatorialtesting}.
For the following metric definitions, let $tCA$ be a $t$-wise covering array of feature model \texttt{fm}. The corresponding equations are shown in Figure~\ref{fig:metrics-summary}.

\begin{metric} \textit{Test Suite Size} is the number of feature sets selected in a covering array for a feature model, shown in Equation~(\ref{eq:suitesize}).
\end{metric}

\begin{metric} \textit{Performance} is the time required for an algorithm to compute a covering array.
\end{metric}

\begin{metric} \textit{Test Suite Similarity}. This metric is defined based on Jaccard's similarity index and applied to variant features. Let $FM$ be the set of all possible feature models, $fs$ and $gs$ be two feature sets in $FS$, and $var:FS \times FM \rightarrow FL$ be an auxiliary function that returns the selected variant features of a feature set according to a FM. The similarity index of two feature sets is thus defined in Equation~(\ref{eq:similarity}), and  the similarity value for the entire covering array is defined by Equation~(\ref{eq:testsuitesimilarity}).

\end{metric}


It should be noted here that the second case of the similarity index, when there are no variant features on both feature sets, is not part of the original proposed comparison framework \cite{DBLP:journals/sqj/PerrouinOSKBT12}. We added this term because in our search we found feature sets formed only with core features.


\begin{metric} \textit{Tuple Frequency}. Let $occurrence:TS \times 2^{FS} \rightarrow \mathbb{N}$  be an auxiliary function that counts the occurrence of a $t$-set (a tuple of t elements) in all the feature sets of a covering array of a single feature model. The metric is defined in Equation~(\ref{eq:tuplefrequency}).
\end{metric}



The first two metrics are the standard measurements used for comparison between different testing algorithms, not only within the SPL domain. To the best of our understanding, the intuition behind the Test Suite Similarity is that the more dissimilar (value close to 0) the feature sets are, the higher chances to detect any faulty behaviour when the corresponding $t$-wise tests are instrumented and performed. Along the same lines, the rationale behind tuple frequency is that by reducing this number, the higher the chances of reducing the repetition of executions of $t$-wise tests.

Let us provide some examples for the latter two metrics for the case of pairwise testing, \texttt{t=2}. Consider for instance, feature sets \texttt{fs0}, \texttt{fs1}, \texttt{fs2 }and \texttt{fs7} from Table \ref{tab:gpl-featuresets}. The variant features in those feature sets are:
\begin{align*}
 var(fs0,gpl)&= \{Undirected, Weight, Prim \} \\
 var(fs1,gpl) &= \{Undirected, Weight, Search, DFS, \\
              & Connected, Kruskal\} \\
 var(fs2,gpl) &= \{Directed, Search, DFS, Number, Cycle\} \\
 var(fs7,gpl) &= \{Undirected, Weight, Search, DFS, \\
              & Connected, Number, Cycle\}
\end{align*}

\begin{sloppypar}
An example is the similarity value between feature sets \texttt{fs0} and \texttt{fs2}, that is  $Sim(fs0,fs2,gpl)=0/8=0.0$. The value is zero because those two feature sets do not have any selected variant features in common.
Now consider $Sim(fs1,fs7,gpl)=5/8=0.625$ which yields a high value because those feature sets have the majority of their selected features in common.
\end{sloppypar}

\begin{sloppypar}
For sake of illustrating the Tuple Frequency metric, let us assume that  the set of feature sets in Table \ref{tab:gpl-featuresets} is a 2-wise covering array of GPL denoted as \texttt{tCA$_{gpl}$}\footnote{There are 24 2-wise pairs, out of the 418 pairs that GPL contains, which are not covered.}. For example, the 2-set
\texttt{ts0} \texttt{=} \texttt{[\{Driver\},\{Prim\}]} is covered by feature sets \texttt{fs1}, \texttt{fs2}, \texttt{fs3}, \texttt{fs4}, \texttt{fs6}, and \texttt{fs7}. Thus, its frequency is equal to $occurrence(ts0,tCA_{gpl})/8 = 6/8 = 0.75$.
As another example, the 2-set \texttt{ts1} \texttt{=} \texttt{[\{Kruskal,DFS\}, $\emptyset$]} is covered by feature set \texttt{fs1}. Thus its frequency is equal to $occurrence(ts1,tCA_{gpl})/8 = 1/8 = 0.125$.
\end{sloppypar}

\begin{figure*}
\caption{Framework Metrics Summary}
\label{fig:metrics-summary}
\normalsize
\hrulefill

\begin{equation}
\label{eq:suitesize}
TestSuiteSize = |tCA|
\end{equation}

\begin{equation}
\label{eq:similarity}
Similarity(fs, gs, fm) =
  \begin{cases}
   \frac{\mid var(fs,fm)~\cap~var(gs,fm) \mid}{\mid var(fs,fm) ~\cup~ var(gs,fm) \mid}  & \text{if }  var(fs,fm)\cup var(gs,fm) \neq \emptyset \\
 0 & \text{otherwise}
  \end{cases}
\end{equation}

\begin{equation}
\label{eq:testsuitesimilarity}
 TestSuiteSimilarity(tCA, fm) = \frac{\sum_{fsi \in tCA} \sum_{fsj \in tCA} Similarity(fsi, fsj, fm)}{\mid tCA \mid^ 2}
\end{equation}

\begin{equation}
\label{eq:tuplefrequency}
TupleFrequency (ts, tCA) = \frac{occurence(ts, tCA)}{\mid tCA \mid}
\end{equation}
\hrulefill
\end{figure*}

Next we present the three algorithms that we used to assess this comparison framework on the  feature models of the proposed benchmark.



\section{Algorithms Overview}
\label{sec:algorithms}
In this section we briefly describe the three testing algorithms we used in our study.

\subsection{CASA Algorithm}
\label{subsec:casa}
CASA is a simulated annealing algorithm that was designed to generate $n$-wise covering arrays for SPLs  \cite{GCD11}. CASA relies on three nested search strategies. The outermost search performs one-sided narrowing, pruning the potential size of the test suite to be generated by only decreasing the upper bound. The mid-level search performs a binary search for the test suite size. The innermost search strategy is the actual simulated annealing procedure, which tries to find a pairwise test suite of size \texttt{N} for feature model \texttt{FM}. For more details on CASA please refer to \cite{GCD11}.

\subsection{Prioritized Genetic Solver}
\label{subsec:geneticsolver}
The \textit{Prioritized Genetic Solver (PGS)} is an evolutionary approach proposed by Ferrer et al.~\cite{DBLP:conf/gecco/FerrerKCA12} that constructs a  test suite taking into account priorities during the generation. PGS is a constructive genetic algorithm that adds one new product to the partial solution in each iteration until all pairwise combinations are covered. In each iteration the algorithm tries to find the product that adds the most coverage to the partial solution. This paper extends and adapts PGS for SPL testing. PGS has been implemented using jMetal~\cite{Durillo2011760}, a Java framework aimed at the development, experimentation, and study of metaheuristics for solving optimization problems. For further details on PGS, please refer to ~\cite{DBLP:conf/gecco/FerrerKCA12}.


\subsection{ICPL}
\label{subsec:icpl}
ICPL is a greedy approach to generate $n$-wise test suites for SPLs, which has been introduced by Johansen \emph{et al.}~\cite{JHF12}\footnote{ICPL stands for \textit{``ICPL Covering array generation algorithm for Product Lines".}}. 
It is basically an adaptation of Chv\'atal's algorithm to solve the set cover problem. 
First, the set $TS$ of all valid $t$-sets that need to be covered is generated. 
Next, the first feature set (product) $fs$ is computed by greedily selecting a subset of $t$-sets in $TS$ that constitute a valid product in the input feature model 
and added to the (initially empty) test suite $tCA$. Henceforth, all $t$-sets that are covered by product $fs$ are removed from $TS$. 
ICPL then proceeds to generate products and adds them to the test suite $tCA$  until $TS$ is empty, i.e. all valid $t$-sets are covered by at least one product. 
To increase ICPLs performance Johansen \emph{et al.} made several enhancements to the algorithm, for instance they parallelized the data independent processing steps. For further details on ICPL please refer to \cite{JHF12}.



\section{Evaluation}
\label{sec:Evaluation}

In this section we present the evaluation of the benchmark using the comparison framework for the case of pairwise testing. We described the statistical analysis performed and the issues found.
All the data and code used in our analysis is available on \cite{vamos14paperURL}.

\subsection{Experimental Set Up}
\label{subsec:experimentalsetup}

\begin{sloppypar}
The three algorithms, CASA, PGS and ICPL, are non-deterministic. For this reason we performed 30 independent runs for a meaningful statistical analysis. All the executions were run in a cluster of 16 machines with Intel Core2 Quad processors Q9400 (4 cores per processor) at 2.66~GHz and 4~GB memory running Ubuntu 12.04.1 LTS and managed by the HT Condor 7.8.4 cluster manager. Since we have 3 algorithms and \nfp feature models the total number of independent runs is $3 \cdot \nfp \cdot 30 = 1,710$.
Once we obtained the resulting test suites we applied the metrics defined in Section~\ref{sec:framework} and we report summary statistics of these metrics.
In order to check if the differences between the algorithms are statistically significant or just a matter of chance, we  applied the Wilcoxon rank-sum~\cite{Sheskin07} test. In order to properly interpret the results of statistical tests, it is always advisable to report effect size measures. For that purpose, we have also used the non-parametric effect size measure $\hat{A}_{12}$  statistic proposed by Vargha and Delaney~\cite{Vargha00}, as recommended by Arcuri and Briand~\cite{Arcuri12}.
\end{sloppypar}

\subsection{Analysis}
\label{subsec:analysis}

The first thing we noticed is the inadequacy of the  Tuple Frequency metric for our comparison purposes. Its definition, as shown in Equation (\ref{eq:tuplefrequency}), applies on a per tuple (i.e. t-set) basis. In other words, given a test suite, we should compute this metric for each tuple and we should provide the histogram of the tuple frequency. This is what the authors of~\cite{DBLP:journals/sqj/PerrouinOSKBT12} do. This means we should show $1,710$ histograms, one per test suite which evidently is not a viable option to aggregate the information of this metric.
An alternative is also presented in~\cite{DBLP:journals/sqj/PerrouinOSKBT12}. It consists on using the average of the tuple frequencies in a test suite taking into account all the tuples. 
Unfortunately we found that this average says nothing about the test suite.
By using counting arguments we show that this average depends only on the number of features and the number of valid tuples of the model (i.e. the average is the same for all the test suites associated to a feature model), hence it is not suited to assess the quality of the test suite. The proof is presented in the Appendix.
In the following we omit any information related to the Tuple Frequency metric and defer to our future work to study how to aggregate it. 


In order to assess whether there was a correlation between the feature metrics, we calculated the Spearman rank's correlation coefficient for each pair of metrics. Table~\ref{tab:correlations} shows the results obtained plus the correlation values with the number of products and the number of features of the FMs (first two columns and  two rows of the table).

\begin{table}[]
\centering
\renewcommand{\tabcolsep}{0.10cm}
\scriptsize
\caption{Spearman's correlation coefficients of all models and algorithms.} \label{tab:correlations}
\begin{tabular}{|l|r|r|r|r|r|}
\hline
 & Products & Features & TSSize & Performance & Similarity \\
\hline
Products & 1 &	0.495 &	0.717 &	0.169 &	-0.015 \\	
Features & 0.495 &	1 &	0.537 &	0.336 &	0.180 \\	
TSSize & 0.717&	0.537&	1 &	0.280&	-0.106 \\
Performance & 0.169	& 0.336	& 0.280 &	1 &	0.440 \\
Similarity & -0.015	& 0.180 &	-0.106 &	0.440 &	1 \\
\hline	
\end{tabular}
\end{table} 	

We can observe a positive and relatively high correlation among the number of products, features and test suite size. This is somewhat expected because the number of valid products is expected to increase when more features are added to a feature model. In the same sense, more features not only usually imply more combinations of features that must be covered by the test suite, but also usually mean that more test cases must be added.

Regarding performance, we expect the algorithms to take more time to generate the test suites for larger models. The positive correlation between the performance and the three previous size-related measures (products, features and size) supports this idea. However, the value is too low (around 0.3 on average) to clearly claim that larger models require more computation time.

The correlation coefficient between the similarity metric and the other  metrics is low except for the performance where the higher the similarity, the longer time spent in building the test suite. In the case of number of products, the value is rather small (-0.015) to draw any conclusions.
In the case of test suite size, we would expect larger test suites to have higher similarity values (positive correlation with test suite size), but this is not the case in general.
We believe these results might be due to the fact that the similarity metric, as defined in Equation~(\ref{eq:testsuitesimilarity}), only considers selected variant features; however, we would argue that unselected features must also be considered in computing similarity because they are also part of the t-sets, which should be covered by the test suites. It is part of our future work to evaluate alternatives to this metric.


Let us now analyse the metrics results grouped by algorithms, shown in Table~\ref{tab:nonprotable}, for the $\nfp$ feature models. In the table we highlight with dark gray the results that are the best for each metric.
The results of the benchmark models reveal that CASA is the best algorithm regarding the size of the test suite (with a statistically significant difference), whereas for PGS and ICPL the differences in test suite size are not statistically significant.
If we focus on performance time, ICPL is the clear winner followed by CASA. PGS is outperformed by CASA in test suite size and computation time. PGS is also the slowest algorithm, in part because it is not specifically designed to deal with feature models.
Regarding the similarity metric, ICPL is the algorithm providing more dissimilar products (with a statistically significant difference) and CASA is the second one, but with no statistically significant difference with PGS.
This is, as we have mentioned before, somewhat counter-intuitive because
ICPL produces larger test suites than CASA, therefore we would expect that test suites with more products to have more similar products considering that the number of features in the feature models is finite. To elucidate why this is the case is part of our future work.

\begin{table}[!ht]
\centering
\renewcommand{\tabcolsep}{0.15cm}
\scriptsize
\caption{Average of the metrics computed on the test suites generated by the three algorithms for the 19 feature models.} \label{tab:nonprotable}
\begin{tabular}{|l|l|r|r|r|r|r|r|r|r|r|c|c|r|}
\hline
Model &	Algor. &	Size &	Performance &	Similarity  \\ \hline
\multirow{3}{*}{Apache} & CASA & \cellcolor{gray95}6.00 & 566.67 & 0.3635 \\
 & PGS & 8.13 & 27196.40 & 0.3450 \\
 & ICPL & 8.00 & \cellcolor{gray95}189.67 & \cellcolor{gray95}0.3044 \\   \hline
\multirow{3}{*}{argo-uml-spl}  & CASA & \cellcolor{gray95}6.00 & 600.00 & 0.3670 \\
 & PGS & 7.97 & 21167.77 & 0.3611\\
 & ICPL & 8.00 & \cellcolor{gray95}321.53 & \cellcolor{gray95}0.3028 \\   \hline
\multirow{3}{*}{ BerkeleyDB}  & CASA & \cellcolor{gray95}6.00 & 12600.00 & 0.3721 \\
 & PGS & 8.00 & 126936.77 & 0.4478 \\
 & ICPL & 7.00 & \cellcolor{gray95}2027.57 & \cellcolor{gray95}0.2679 \\   \hline
\multirow{3}{*}{ BerkeleyDBF}  & CASA & \cellcolor{gray95}6.00 & 533.33 & 0.3645 \\
 & PGS & 8.10 & 28591.70 & 0.3515 \\
 & ICPL & 7.00 & \cellcolor{gray95}176.77 & \cellcolor{gray95}0.3232 \\   \hline
\multirow{3}{*}{ BerkeleyDBM}  & CASA & \cellcolor{gray95}30.00 & 7333.33 & 0.3482 \\
 & PGS & 30.70 & 249834.47 & 0.3597 \\
 & ICPL & 31.00 & \cellcolor{gray95}554.47 & \cellcolor{gray95}0.2684 \\   \hline
\multirow{3}{*}{ BerkeleyDBP}  & CASA &\cellcolor{gray95} 9.53 & 4133.33 & 0.3826 \\
 & PGS & 11.43 & 65988.83 & 0.3915 \\
 & ICPL & 10.00 & \cellcolor{gray95}366.03 & \cellcolor{gray95}0.3576 \\   \hline
\multirow{3}{*}{ Curl}  & CASA & \cellcolor{gray95}8.00 & 916.67 & 0.3537 \\
 & PGS & 12.13 & 43605.03 & 0.3490 \\
 & ICPL & 12.00 & \cellcolor{gray95}276.03 & \cellcolor{gray95}0.2634 \\   \hline
\multirow{3}{*}{ DesktopSearcher}  & CASA & \cellcolor{gray95}8.37 & 2266.67 & \cellcolor{gray95}0.3785 \\
 & PGS & 9.20 & 32552.70 & 0.3951 \\
 & ICPL & 9.07 & \cellcolor{gray95}412.20 & 0.3895 \\   \hline
\multirow{3}{*}{ fame-dbms-fm}  & CASA & \cellcolor{gray95}10.00 & 1700.00 & 0.3255 \\
 & PGS & 13.80 & 58227.27 & 0.3327 \\
 & ICPL & 11.00 & \cellcolor{gray95}378.00 & \cellcolor{gray95}0.3190 \\   \hline
\multirow{3}{*}{ gpl}  & CASA & \cellcolor{gray95}12.00 & 1966.67 & 0.3494 \\
 & PGS & 13.13 & 62859.50 & 0.3576 \\
 & ICPL & 13.00 & \cellcolor{gray95}317.83 & \cellcolor{gray95}0.3481 \\   \hline
\multirow{3}{*}{ LinkedList}  & CASA & \cellcolor{gray95}12.13 & 2133.33 & 0.4057 \\
 & PGS & 15.43 & 74601.10 & 0.4151 \\
 & ICPL & 14.00 & \cellcolor{gray95}462.53 & \cellcolor{gray95}0.3988 \\   \hline
\multirow{3}{*}{ LLVM}  & CASA & \cellcolor{gray95}6.00 & 633.33 & 0.3653 \\
 & PGS & 8.73 & 32615.13 & 0.3523 \\
 & ICPL & 9.00 & \cellcolor{gray95}221.73 & \cellcolor{gray95}0.2320 \\   \hline
\multirow{3}{*}{ PKJab}  & CASA & \cellcolor{gray95}6.00 & 550.00 & 0.3752 \\
 & PGS & 7.63 & 28318.90 & 0.3726 \\
 & ICPL & 7.00 & \cellcolor{gray95}193.13 &\cellcolor{gray95} 0.3439 \\   \hline
\multirow{3}{*}{ Prevayler}  & CASA & \cellcolor{gray95}6.00 & 550.00 & 0.3610 \\
 & PGS & 6.60 & 26052.00 & 0.3598 \\
 & ICPL & 8.00 & \cellcolor{gray95}156.20 & \cellcolor{gray95}0.2677 \\   \hline
\multirow{3}{*}{ SensorNetwork}  & CASA & \cellcolor{gray95}10.23 & 1583.33 & 0.3719 \\
 & PGS & 12.30 & 63212.33 & 0.3595 \\
 & ICPL & 14.00 & \cellcolor{gray95}445.27 & \cellcolor{gray95}0.3166 \\   \hline
\multirow{3}{*}{ TankWar}  & CASA & \cellcolor{gray95}12.50 & 39200.00 & 0.3483 \\
 & PGS & 14.77 & 152199.10 & 0.3571 \\
 & ICPL & 14.00 & \cellcolor{gray95}612.23 & \cellcolor{gray95}0.3140 \\   \hline
\multirow{3}{*}{ Wget}  & CASA &\cellcolor{gray95} 9.00 & 766.67 & 0.3548 \\
 & PGS & 12.43 & 46869.77 & 0.3541 \\
 & ICPL & 12.00 &\cellcolor{gray95} 290.87 & \cellcolor{gray95}0.2685 \\   \hline
\multirow{3}{*}{ x264}  & CASA & \cellcolor{gray95}16.00 & 2966.67 & 0.3523 \\
 & PGS & 16.97 & 74500.10 & 0.3640 \\
 & ICPL & 17.00 & \cellcolor{gray95}359.53 & \cellcolor{gray95}0.2574 \\   \hline
\multirow{3}{*}{ ZipMe}  & CASA & \cellcolor{gray95}6.00 & 533.33 & 0.3505 \\
 & PGS & 7.43 & 26376.53 & 0.3748 \\
 & ICPL & 7.00 & \cellcolor{gray95}165.40 & \cellcolor{gray95}0.3429 \\   \hline

\end{tabular}
\end{table}

Finally, in Table~\ref{tab:a12algorithms} we show the $\hat{A}_{12}$ statistic to assess the practical significance of the results. Given a performance measure $M$, $\hat{A}_{12}$ measures the probability that running algorithm \emph{A} yields higher $M$ values than running another algorithm \emph{B}. If the results of the two algorithms are equal, then $\hat{A}_{12}=0.5$. If $\hat{A}_{12}=0.3$ entails one would obtain higher values for $M$ with algorithm \emph{A}, 30\% of the times.
In Table~\ref{tab:a12algorithms} we have highlighted the largest distance from 0.5 (equality) per quality indicator, note that 0.5 indicates no difference in the comparison. Recall that we have to highlight two values per metric, because the direction of the comparison does not affect the interpretation of the result, although the value itself is complementary (both adding up to 1).

Regarding size, there is not statistically significant difference between ICPL and PGS, while CASA obtains the best results in more than 66\% of the times.
Regarding performance, ICPL is faster with a higher probability than the other algorithms. ICPL is faster than CASA in 93\% of the times, moreover, ICPL is faster than PGS in 99.95\% of the times. So, ICPL is clearly the best algorithm in performance without any doubts.
Regarding similarity, ICPL is again the algorithm which obtains more dissimilar test suites. It obtains a lower value of similarity than CASA and PGS, in around 85\% of the comparisons.
As we have commented earlier in this section, this results of similarity are somehow unexpected, because smaller test suites ought to be more dissimilar than larger ones. For this reason, CASA would obtain lower values of test suite similarity, but it does not. Again, investigating why  this is the case is part of our future work.
\begin{table}[]
\renewcommand{\tabcolsep}{0.05cm}
\centering
\tiny
\begin{tabular}{|c|c|c|c|c|c|c|c|c|c|c|c|c|c|c|c|c|}
\hline
 & \multicolumn{3}{|c|}{Size} & \multicolumn{3}{|c|}{Performance} & \multicolumn{3}{|c|}{Similarity} \\ \hline
    & CASA     & ICPL &      PGS &     CASA &    ICPL &     PGS & CASA & ICPL & PGS \\ \hline
CASA & - &  0.3312 &  \cellcolor{gray95}0.3194   &  -    &  0.9286  &  0.0109  &   -   &  \cellcolor{gray95}0.8479    & 0.4807  \\ \hline
ICPL  & 0.6688 & - &  0.4653  &  0.0714   &  -   &  \cellcolor{gray95}0.0005 &   \cellcolor{gray95}0.1521   &  -   &  0.1577     \\ \hline
PGS  & \cellcolor{gray95}0.6806  &  0.5347  &   -  &  0.9891    &  \cellcolor{gray95}0.9995   &  -   &  0.5193    &  0.8423   &   -    \\ \hline
\end{tabular}
\caption{$\hat{A}_{12}$ statistical test results.}
\label{tab:a12algorithms}
\end{table}


\section{Related Work}
\label{sec:relatedwork}

There exists substantial literature on SPL testing \cite{DBLP:journals/infsof/EngstromR11, DBLP:journals/infsof/NetoMMAM11, DBLP:conf/splc/LeeKL12,DBLP:journals/sigsoft/MachadoMA12}. However, to the best of our knowledge there are neither benchmarks nor frameworks for comparing approaches. 
In the area of Search-Based Software Engineering a major research focus has been software testing ~\cite{DBLP:journals/csur/HarmanMZ12,DBLP:conf/ssbse/FreitasS11}, where there exists a plethora of articles that compare testing algorithms using different metrics.  For example, Mansour \emph{et al.}~\cite{Mansour2001} compare five algorithms for regression testing using eight different metrics (including quantitative and qualitative criteria). Similarly, Uyar et al.~\cite{Uyar2006} compare different metrics implemented as fitness functions to solve the problem of test input generation.
To the best of our knowledge, in the literature on test case generation there is no well-known comparison framework for the research and practitioner community to use. Researchers usually apply their methods to open source programs and compute some metrics directly such as the success rate, the number of test cases and performance. The closest to a common comparison framework we could trace is the work of Rothermel and Harrold~\cite{Rothermel1994} where they  propose a framework for regression testing.


\section{Conclusions and Future Work}
\label{sec:conclusions}

In this research-in-progress paper, we put forward $\nfp$ feature models as a basis for a benchmark of CIT testing of SPLs. With this benchmark, we made an assessment of the comparison framework proposed by Perrouin \emph{et al.} using three approaches (CASA, PGS and ICPL) for the case of pairwise testing.
Overall the framework helped us identify facts such as that CASA obtains the smallest test suites, while ICPL is the fastest algorithm and also obtains the most dissimilar products. However, we also identified two shortcomings of this framework: \textit{i)} similarity does not consider features that are not selected in a product, a fact that might skew the expected output, and \textit{ii)} tuple frequency is applicable on a per tuple basis only, so its value as an aggregative measure is not clear.
As future work we plan to evaluate other metrics that could be used to complement the framework, for this we will follow the guidelines for metrics selection suggested in \cite{DBLP:journals/tosem/MeneelySW12}. In addition, we expect to integrate more feature models into the benchmark as well as to refine or extend the feature model selection criteria.



\section{Acknowledgements}

This research is partially funded by the Austrian Science Fund (FWF) project P21321-N15 and Lise Meitner Fellowship M1421-N15, the Spanish Ministry of Economy and Competitiveness and FEDER under contract TIN2011-28194.
We thank Martin Johansen and {\O}ystein Haugen for their help with SPLCA tool, Norbert Siegmund for his support with SPLConqueror, and  and B.J. Garvin for his help with CASA.


\section{Appendix}

Let $ocurrence_{p}(ts, fs)$ compute the number of times (i.e. 0 or 1) that $t$-set $ts$ appears in feature set $fs$.  Thus we define $ocurrence (ts, tCA)$ as:
\begin{align*}
ocurrence(ts, tCA) = \sum_{fs \in tCA} ocurrence_p(ts,fs)\nonumber
\end{align*}

\begin{theorem}
The average tuple frequency depends only on the number of features, $|FL|$, and valid $t$-sets, $|TS|$.
\end{theorem}
\begin{proof}
We can write the average tuple frequency as:
\begin{align*}
\frac{1}{|TS|}&\sum_{ts \in TS} \frac{ocurrence(ts,tCA)}{|tCA|} \nonumber \\ 
&= \frac{1}{|TS|\cdot |tCA|} \sum_{ts \in TS} ocurrence(ts,tCA) \nonumber \\
&= \frac{1}{|TS|\cdot |tCA|}\sum_{ts \in TS} \sum_{fs \in tCA} ocurrence_p(ts,fs) \nonumber\\
&=  \frac{1}{|TS|\cdot |tCA|}\sum_{fs \in tCA} \left( \sum_{ts \in TS}  ocurrence_p(ts,fs)\right)  \nonumber \\
&= \frac{1}{|TS|\cdot |tCA|}\sum_{fs \in tCA} \frac{|FL|(|FL|-1)}{2}  \nonumber \\
&= \frac{|tCA|}{|TS|\cdot |tCA|} \frac{|FL|(|FL|-1)}{2} \nonumber\\
&= \frac{|FL|(|FL|-1)}{2|TS|}
\end{align*}

Note the expression within the parentheses. This is the number of valid tuples in a feature set $fs$. Let us now select two arbitrary features from $FL$. These features can be both selected in $fs$, both unselected or one selected and the other not. In any case, as the product is a valid product, there exists a valid $t$-set in $TS$ having these two arbitrary features that is covered by $fs$. Thus, the sum within the parentheses is the number of pairs of features: 
\begin{align*}
\sum_{ts \in TS}  ocurrence_p(ts,fs)= \frac{|FL|(|FL|-1)}{2}\nonumber
\end{align*}
\end{proof}




\bibliographystyle{abbrv}
\bibliography{biblio}

\end{document}